\documentclass[a4paper]{article}

\usepackage{graphicx}
\usepackage[margin=1in]{geometry}
\usepackage[english]{babel}

\usepackage{amsthm}
\usepackage{amssymb,amsmath}  
\usepackage{mathtools}
\usepackage{graphicx}
\usepackage{hyperref}
\usepackage[capitalize]{cleveref}
\usepackage{xspace}
\usepackage{thmtools}
\usepackage{thm-restate}
\usepackage{url}
\usepackage{caption}
\usepackage{subcaption}
\usepackage{authblk}
\usepackage{apxproof}

\usepackage[ruled,noend]{algorithm2e}
\crefname{algocf}{algorithm}{algorithms}

\usepackage{enumitem}
\setlist[itemize]{noitemsep, nolistsep}
\setlist[enumerate,1]{label=(\arabic*)}

\def\dd{\mathinner{.\,.}} 
\newcommand{\cO}{\mathcal{O}}

\newcommand{\N}{\mathbb{N}}
\newcommand{\R}{\mathbb{R}}
\newcommand{\howfar}{\varepsilon}

\newcommand{\setsize}[1]{\left|#1\right|}
\DeclareMathOperator{\qpSet}{\mathrm{QP}_\Sigma}
\DeclareMathOperator{\selfoverlap}{\mathrm{PS}}

\newcommand{\pre}{\textrm{pre}}
\newcommand{\post}{\textrm{post}}
\newcommand{\yes}{\textsf{YES}}
\newcommand{\no}{\textsf{NO}}

\usepackage{marginnote}

\newtheorem{theorem}{Theorem}

\newtheorem{definition}[theorem]{Definition}
\newtheorem{lemma}[theorem]{Lemma}

\newtheorem{corollary}[theorem]{Corollary}
\newtheorem{observation}[theorem]{Observation}
\newtheorem{example}[theorem]{Example}

\crefname{observation}{Observation}{Observations}

\author[1]{Christine Awofeso}
\author[2,3]{Ben Bals}
\author[1]{Oded Lachish}
\author[2,3]{Solon P.\ Pissis}
\affil[1]{Birkbeck, University of London, London, UK}
\affil[2]{CWI, Amsterdam, The Netherlands}
\affil[3]{Vrije Universiteit, Amsterdam, The Netherlands}

\begin{document}

\title{Testing Quasiperiodicity}

\maketitle

\begin{abstract}
A cover (or quasiperiod) of a string $S$ is a shorter string $C$ such that every position of $S$ is contained in some occurrence of $C$ as a substring. The notion of covers was introduced by Apostolico and Ehrenfeucht over 30 years ago [Theor.~Comput.~Sci.~1993] and it has received significant attention from the combinatorial pattern matching community. In this note, we show how to efficiently test whether $S$ admits a cover. 
Our tester can also be translated into a streaming algorithm.
\end{abstract}

\section{Introduction}

A cover (or quasiperiod) of a string $S$ is a shorter string $C$ such that every position of $S$ is contained in some occurrence of $C$ as a substring. The notion of covers generalizes the notion of period. It was introduced by Apostolico and Ehrenfeucht in 1993~\cite{DBLP:journals/tcs/ApostolicoE93}, and since then it has received a lot of attention from the combinatorial pattern matching community.
For example, the shortest cover of a string of length $n$ can be computed in $\cO(n)$ time~\cite{DBLP:journals/ipl/ApostolicoFI91,DBLP:journals/ipl/Breslauer92}; see~\cite{DBLP:journals/tcs/CzajkaR21,DBLP:journals/fuin/MhaskarS22} for surveys.

Our main result here is a tester to determine whether a string $S$ of length $n$ has a cover of length at most $q$ or the minimum Hamming distance of $S$ and a string that has such a cover is at least $\howfar n$, for some small $\howfar \in \R^+$. The tester does not access $S$ directly and instead uses queries to an oracle of the form: \emph{what is the letter at position $i\in [n]$ of $S$?}
Our algorithm uses $\cO(q^3 \howfar^{-1}\log q )$ such queries, which is independent of $n$; see \cref{sec:tester}. Notably, our combinatorial insights yield a simple streaming algorithm for short covers; see \cref{sec:stream}. We start with \cref{sec:prel}, which provides the necessary notation, definitions, and tools.
Our work proceeds along the lines of~\cite{lachish_testing_2011}, where the authors provide testers for periodicity.

\section{Preliminaries}\label{sec:prel}

We consider finite strings on an integer \emph{alphabet} $\Sigma=[\sigma]=\{1,2,\ldots,\sigma\}$. 
The elements of $\Sigma$ are called \emph{letters}.
For a string $S = S[1]\cdots S[n]$ on $\Sigma$,
its \emph{length} is $|S| = n$. 
For any $1 \leq i \leq j \leq n$, the string $S[i] \cdots S[j]$ is called a \emph{substring} of $S$. By $S[i\dd j]$, we denote its occurrence at the (starting) position $i$, and we call it a \emph{fragment} of $S$. 
When $i = 1$, this fragment is called a
\emph{prefix}, and when $j = n$, it is called a \emph{suffix}.
The \emph{Hamming distance} of two strings $S, S'\in \Sigma^n$ is $|\{i \in [n] : S[i] \neq S'[i]\}|$.
An integer $p$, $1 \leq p \leq n$, is a \emph{period} of a string $S$ if $S[i] = S[i + p]$, for all $1 \leq i \leq |S|-p$. We say that $B$ is a \emph{border} of $S$ if it is a prefix and a suffix of $S$.
 The notion of cover generalizes the notion of period.

\begin{definition}[Cover and Quasiperiod]
    \label{def:cover}
    A string $C$ is a \emph{cover} of a string $S$ if every position in $S$ lies within an occurrence of $C$ as a substring; that is, for all $i \in [|S|]$, there is an occurrence of $C$ in $S$ that starts at one of $\max(i - |C| +1, 1), \dd, i$. If $C$ is a cover of $S$, we say that $S$ has a \emph{quasiperiod} $\setsize{C}$. 
\end{definition}

\begin{example}
$C=\texttt{aba}$ and $C'=\texttt{abaaba}$ are covers of $S=\texttt{abaababaababaaba}$.
\end{example}

For any $q\in \N$, by $\qpSet(q)$ we denote the set of all strings on $\Sigma$ with a quasiperiod at most $q$.
Since every cover of a string must be a border, any string $S\in\Sigma^n$ has $\cO(n)$ covers, and, in fact, all covers of $S$ can be computed in $\cO(n)$ time~\cite{DBLP:journals/ipl/MooreS94,DBLP:journals/ipl/MooreS95}.
The notion of seed~\cite{DBLP:journals/algorithmica/IliopoulosMP96} generalizes the notion of cover.

\begin{definition}[Seed]
A string $C$ is a \emph{seed} of $S$ if $|C| \leq |S|$ and $C$ is a cover of some string containing $S$ as a substring.
\end{definition}

\begin{example}
$C=\texttt{aba}$ and $C'=\texttt{abaab}$ are seeds of $S=\texttt{aabaababaababaabaa}$.    
\end{example}

Unlike covers, the number of distinct seeds of $S\in\Sigma^n$ can be $\Theta(n^2)$~\cite{DBLP:journals/talg/KociumakaKRRW20}.
\cref{thm:fast-seeds} allows checking whether any fragment of $S$ is a seed efficiently.

\begin{theorem}[\cite{DBLP:conf/soda/KociumakaKRRW12,DBLP:conf/esa/Radoszewski23}]
\label{thm:fast-seeds}
An $\cO(n)$-size representation of all seeds of a string $S\in\Sigma^n$ can be computed in $\cO(n\log n)$ time and $\cO(n)$ space. If $\Sigma=[\sigma]$, with $\sigma=n^{\cO(1)}$, the same representation can be computed in $\cO(n)$ time.    
\end{theorem}

We use simple tools from Diophantine number theory in our analysis.

\begin{definition}[Conical Combination]
A \emph{conical combination} of the natural numbers $a_1, \dots, a_k$ is a number $n=x_1 a_1 + \dots, x_k a_k$, for some $x_1, \dots, x_k \in \N$. 
\end{definition}

This well-known result of Erd\H{o}s and Graham underlies our algorithm.%
\footnote{In particular, the result for 6, 9, and 20 is famous as the Chicken McNugget theorem.}

\begin{theorem}[Frobenius Number Bound, \cite{erdos1972linear}]\label{thm:Frobenius}
Let $a_1 < \dots < a_k \in \N^+$ be set-wise co-prime (i.e., $\gcd(a_1, \dots, a_k) = 1$). Any number $n > 2a_{k-1} \lfloor a_k/k\rfloor - a_k$ can be written as $n = x_1 a_1 + \dots + x_k a_k$, for some $x_1, \dots, x_k \in \N$.
\end{theorem}

Note that \cref{thm:Frobenius} does not require the numbers to be pairwise co-prime.

\begin{corollary}[\cite{erdos1972linear}]
    \label{cor:square-frobenius-bound}
    Let $A \subseteq \N^+$ be bounded by $q \in \N$ and assume $\gcd(A) = 1$.
    Then any number $n \ge 2q^2$ can be written as a conical combination of $A$.
\end{corollary}

\begin{corollary}
    \label{cor:all-conical-combinations}
    Let $A \subseteq \N^+$ be bounded by $q \in \N$.
    Then any number $n \ge 2q^3$ such that $\gcd(A) | n$ can be written as a conical combination of $A$.
\end{corollary}
\begin{proof}
     Apply \cref{cor:square-frobenius-bound} to $A' \coloneqq \{a / \gcd(A) \mid a \in A\}$ and $n' \coloneqq n / \gcd(A)$.
     Multiply the resulting conical combination by $\gcd(A)$.
\end{proof}

\begin{definition}[$q$-Cover Tester]
A \emph{$q$-cover tester} is a randomized algorithm that receives as input $q,n\in \N$ and $\howfar \in \R^+$ and has oracle access to a string $S\in\Sigma^n$. It returns \yes~if $S$ has a quasiperiod at most $q$ and \no~with probability at least $3/4$ if $S$ is \emph{$\howfar$-far} from having a quasiperiod at most $q$; i.e., the minimum Hamming distance of $S$ and a string $S'$ that has such a cover is at least $\howfar n$.
\end{definition}
A $q$-cover tester does not access the input string $S$ directly; instead, it uses queries to the oracle. A \emph{query} is an integer $i\in [n]$ provided to the oracle, on which the oracle returns the letter $S[i]$.   
The \emph{query complexity} of a tester is the maximum number of queries it uses as a function of the parameters $q$, $n$, and $\howfar$.

\section{An Efficient Tester for Covers and Seeds}\label{sec:tester}

Let us fix a string $S \in \Sigma^n$ and a string $C \in \Sigma^q$, for two integers $0<q<n$.
We wish to establish results that help us test whether $C$ is a cover of $S$.

\begin{observation}
    \label{obs:superstring-cover-size}
    Let $C$ be a seed of $S$. Then there is a string $S'=X\cdot S \cdot Y$, with $|X|\in [0,q]$ and $|Y|\in [0,q]$, such that $C$ is a cover of $S'$.
\end{observation}

\begin{observation}
    \label{obs:seed-of-substrings}
If $C$ is a cover of $S$, then $C$ is a seed of any substring of $S$ whose length is at least $|C|$.
\end{observation}

\begin{lemma}
    \label{lem:cover-combination}
    Let $S_1=S[i\dd j]$ and $S_2=S[i'\dd j']$ be two fragments of $S$, with $i\leq i'$ and $j\leq j'$, so that they (1) share at least $2q$ positions of $S$ and (2) both have $C$ as a seed.
    Let $S_3 \coloneqq S[i\dd j']$.
    Then $C$ is also a seed of $S_3$.
\end{lemma}

\begin{figure}[t]
    \centering
    \includegraphics[page=1,width=9cm]{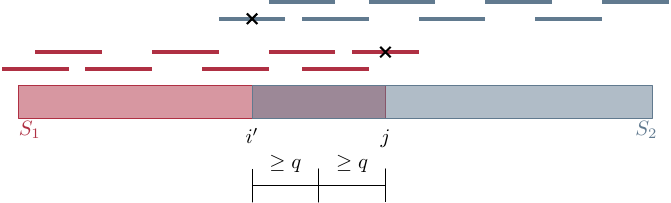}
    \caption{Combining the seed occurrences of two fragments with a long overlap.}
    \label{fig:overlap-combination}
\end{figure}

\begin{proof}
    To construct a covering of $S_3$ with seed $C$, take such coverings for $S_1$ and $S_2$; 
    see \Cref{fig:overlap-combination} for an illustration.
    From the covering of $S_1$, remove the occurrences of $C$ that start after $j-q$. 
    From the covering of $S_2$, remove the occurrences of $C$ that start before $i'$.
    Since $j-i'+1 \ge 2q$, we have that the union of these coverings now covers all of $S_3$. 
    Thus, $C$ is a seed of $S_3$.
\end{proof}

\begin{definition}[Period set]
    We define the \emph{period set} of $C$, denoted by $\selfoverlap(C)$, as the set of periods of $C$. Thus, $\selfoverlap(C)\subseteq [|C|]$.
\end{definition}

We are interested in the ways we can combine copies of $C$ to form a longer string.
The trivial way is concatenating $C$ with itself: the length of $C$ is a period of $C$. The period set tells us which other ways are possible:
for every $\lambda \in \selfoverlap(C)$, we can construct a string of length $\setsize{C} + \lambda$ by overlapping $C$ with itself.

\begin{lemma}
    \label{lem:gcd-string-construction}
    Let $\ell \ge 2q^3+q$ such that $\gcd(\selfoverlap(C))|\ell$. There is a string of length $\ell$ for which $C$ is a cover.
\end{lemma}

\begin{proof}
    We construct such a string, denoted by $S'$, by overlapping $C$ with itself.

    Let $\selfoverlap(C) = \{a_1, \dots, a_k\}$ and let $x_1, \dots, x_k \in \N$ be such that $\sum_{i=1}^k x_i a_i = \ell - q$.
    These coefficients exist by \Cref{cor:all-conical-combinations}.
    Start with $S':=C$.
    Then for each $i \in [k]$, add a copy of $C$, overlapping by $(q-a_i)$ positions.
    This is possible by the definition of \(\selfoverlap(C)\).
    This extends the string by $a_i$ letters.
    Repeat this $x_i$ times.
    At the end of this process, $S'$ is coverable by $C$ and has the correct length.
\end{proof}

The following lemma proves that when $C$ is a cover of $S$, the period set $\selfoverlap(C)$ determines all the possible positions at which $C$ can occur in $S$.

\begin{lemma}
    \label{lem:gcd-of-positions}
    Let $C$ be a cover of $S$ and $P := \{p_1 < \dots < p_k\}$ the positions at which $C$ occurs in $S$.  
    Let $P' \coloneqq \{p_i - 1 \mid i \in [k]\}$.
    Then $\gcd(\selfoverlap(C)) | \gcd(P')$.
\end{lemma}

\begin{proof}
    Let $p_i' \in P'$. 
    We will show that $\gcd(\selfoverlap(C)) | p'_i$, which implies the claim.

    Fix an arbitrary covering of $S$ using $C$.
    Let $e_i$ be the ending position of the occurrence of $C$ in this covering that ends first on or after position $p'_i$.
    Then $C$ is a cover of the string $S[1\dd e_i]$, and therefore, by the definition of $\selfoverlap(C)$, there must be $\{x_a \in \N\}_{a \in \selfoverlap(C)}$ such that $e_i = \sum_{a \in \selfoverlap(C)} x_a a$.
    Since there is an occurrence of $C$ in $S$ at positions $e_i$ and $p_i$, 
    and $e_i - p_i \le q$ (by choice of $e_i$), we have $e_i - p_i - 1 \in \selfoverlap(C)$ and thus $e_i - p_i' \in \selfoverlap(C)$. 
    Therefore, $p_i'$ can be written as a conical combination of $\selfoverlap(C)$.
    This implies that $\gcd(\selfoverlap(C)) | p_i'$.
\end{proof}

\begin{definition}[$C$-Consistent Fragment]\label{def:consistent}
    A fragment $S[i \dd j]$ of $S$ is \emph{$C$-consistent} if and only if (1) $C$ is a seed of $S[i\dd j]$ and (2) if $p\in[i,j]$ is the position of an occurrence of $C$ in $S[i\dd j]$, then $\gcd(\selfoverlap(C)) | p-1$.
\end{definition}

Let $\mathcal S = \{S_i\}_{i \in [n / (2q^3)]}$ be the set of fragments of $S$ obtained by splitting $S$ into fragments of length $4q^3$ (the last fragment may be shorter than $4q^3$ or empty), each overlapping by $2q^3$ positions.
Thus, we have $|\mathcal S|=\cO(n/q^3)$. 

An immediate consequence of \cref{def:consistent} is that if a fragment of $S$ is not $C$-consistent, then $C$ is not a cover of $S$. \cref{lem:size} shows that if $S$ is $\howfar$-far from $\qpSet(q)$, then this can be detected with high probability by picking a random fragment from set $\mathcal S$, which, as we explain next, is what our algorithm does.

\begin{lemma}\label{lem:size}
    If $S\in\Sigma^n$ is $\howfar$-far from $\qpSet(q)$, for some fixed $C\in\Sigma^q$, the set $\mathcal S' \subseteq \mathcal S$ of $C$-consistent fragments of $S$ has size at most $(1-\howfar)\cdot n / (2q^3)$.
\end{lemma}

\begin{proof}
    Let $\mathcal T$ be the subset of $\mathcal S$ that includes all fragments that are not $C$-consistent. 
    Assume for the sake of contradiction that $|\mathcal{T}| \leq \howfar n/(2q^3)$.
    We will show that this implies that $S$ is at Hamming distance less than $\howfar n$ from $\qpSet(q)$, that is, that $S$ is not $\howfar$-far from $\qpSet(q)$.
    In particular, we will show that we can edit $S$ only in fragments that are in the set $\mathcal T$ so that $C$ becomes a cover of $S$.
    Note that there are at most $\howfar n$ positions in $\mathcal T$ by the assumption on its size.

    Consider any maximal sequence $S_1, \dots, S_k$ of consecutive fragments in $\mathcal T$. 
    Note that $k$ could be equal to one.
    Let $S_\pre$ be the fragment in $\mathcal S$ before $S_1$ and $S_\post$ the one after $S_k$.
    As we assumed $S_1, \dots, S_k$ to be maximal, we can assume $S_\pre, S_\post \in \mathcal S'$ and thus also that $C$ is a seed of $S_\pre$ and $S_\post$.
    Fix coverings of $S_\pre$ and $S_\post$ with seed $C$ (see \Cref{fig:fixing}).
    For these coverings, let $i_\pre$ be the first position in $S$ after the covering of $S_\pre$ and $i_\post$ the last position before the covering of $S_\post$. 
    By \Cref{obs:superstring-cover-size}, these are at most $q$ positions before or after the end or start of $S_\pre$ and $S_\post$, respectively.
    The fragment $S[i_\pre\dd i_\post]$ can be edited to be coverable by $C$ by replacing it with the string given by \Cref{lem:gcd-string-construction}.

    \begin{figure}[tbp]
    \centering
    \includegraphics[page=2,width=9cm]{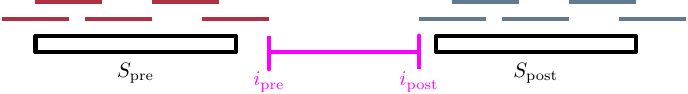}
    \caption{Editing a string that is not $\howfar$-far from $\qpSet(q)$ to be in $\qpSet(q)$.}
    \label{fig:fixing}
\end{figure}
    
    Repeating this for any maximal sequence $S_1, \dots, S_k$ of consecutive fragments in $\mathcal{T}$, yields a string that is coverable by $C$.
    The coverings of the individual fragments combine to yield a covering of the entire string $S$ by \Cref{lem:cover-combination}.
\end{proof}

\begin{theorem}
    \Cref{alg:main} is a tester deciding whether a string $S\in\Sigma^n$ has a cover of length at most $q$ using $\cO(q^3 \log q \cdot \howfar^{-1})$ queries, for some $\howfar \in \R^+$. 
\end{theorem}

\begin{proof}
    The query complexity is immediate from the algorithm's definition.

    For the correctness, first observe that if $S$ has a cover $C$ of length at most $q$, then any set of fragments of $S$ is $C$-consistent (\Cref{def:consistent}).
    The first condition follows from \Cref{obs:seed-of-substrings} and the second one from \Cref{lem:gcd-of-positions}.

    \begin{algorithm}[b]
    \DontPrintSemicolon 
    \textbf{Input: } $q \in \N, n \in \N, \howfar \in \R^+$ and oracle access to a string $S\in\Sigma^n$  
    
    \textbf{Output: } $b \in \{\no, \yes\}$
    
    Let set $\mathcal S := \{S_i\}_{i \in [n / (2q^3)]}$ be defined as above.

    Sample $(24 \log q) / \varepsilon$ of the length-$(4q^3)$ fragments in $\mathcal S$ uniformly at random
and also the length-$(4q^3)$ suffix of $S$. Denote the sampled fragments by set $\mathcal R$.

    Query the $\cO(q^3 \log q \cdot \howfar^{-1})$ positions of the fragments in $\mathcal R$ using the oracle.

    Query the $\cO(q)$ positions $1, \dots, q$ and $n-q+1, \dots, n$ using the oracle.

    For every border $C$ of $S$, $|C|\leq q$, check whether every $F\in \mathcal R$ is $C$-consistent.
    
    If there is such a border $C$, output $\yes$; otherwise output $\no$.

    \caption{A $q$-cover tester\label{alg:main}}
\end{algorithm}

    For the other direction of the correctness, assume that $S$ is $\howfar$-far from $\qpSet(q)$.
    There are $q$ prefixes of $S$ that could be a cover of length at most $q$.
    For any such prefix $C$, the set of fragments in $\mathcal S$ that are $C$-consistent has size at most $(1-\varepsilon) \cdot n / (2|C|^3)$ by \cref{lem:size}.
    Since we sample $(24 \log q) / \varepsilon$ of these, we will, with probability at least $3/4$, reject all possible covers.
    This follows from a simple union bound over the at most $q$ candidate borders, yielding the desired result.
\end{proof}

\Cref{alg:main} can be adapted to check if $S$ has a seed by considering the $\cO(q^2)$ fragments of $S[1 \dd 2q]$ as candidate seeds. Note that, since $2q=\cO(q)$, this does not affect the query complexity of the tester. Additionally, the requirement that the first and last fragments in $\mathcal S$ must be covered should be slightly relaxed. 

\begin{corollary}
    There is a tester deciding whether a string $S\in\Sigma^n$ has a seed of length at most $q$ using $\cO(q^3 \log q \cdot \howfar^{-1})$ queries, for some $\howfar \in \R^+$.
\end{corollary}

\section{A Simple Streaming Algorithm for Covers via Seeds}\label{sec:stream}

Gawrychowski et al.~showed
a one-pass streaming algorithm for computing the shortest cover of a string of length $n$ that uses $\cO(\sqrt{n \log n})$ space and runs in $\cO(n\log^2 n)$ time~\cite{DBLP:conf/cpm/GawrychowskiRS19}. 
One of its routines is a streaming algorithm for computing the length of the shortest cover if it is at most $q$ that uses $\cO(q)$ space and runs in $\cO(n)$ time. 
The algorithm underlying \cref{thm:stream} is a fundamentally different and \emph{simple} alternative for the same computation that also uses $\cO(q)$ space and runs in $\cO(n)$ time. It follows from our results in \cref{sec:tester} and can be seen as a different, independently useful consequence of our combinatorial insights. 

\begin{theorem}\label{thm:stream}
There is a one-pass streaming algorithm for computing the shortest cover $C$ of $S\in\Sigma^n$, if $|C|\leq q$, that uses $\cO(q)$ space and runs in $\cO(n)$ time. 
\end{theorem}

\begin{proof}
    We conceptually split $S$
    into a set $\mathcal S'$ of $\cO(n/q)$ fragments of $S$ of length $4q$, each overlapping by $2q$ positions (the last fragment may be shorter).
    Then, by \Cref{lem:cover-combination} and \Cref{obs:seed-of-substrings}, we have that any border $C$ of length at most $q$ of $S$ is a cover of $S$ if and only if $C$ is a seed of all fragments in $\mathcal S'$.
    
    This fact yields a simple streaming algorithm.
    We start by reading $P:=S[1\dd q]$ in memory. The prefixes of $P$ will be our \emph{candidates}.
    We also insert these $q$ letters in $\cO(q)$ total time in a dynamic dictionary $\mathcal{D}$ supporting $\cO(1)$-time worst-case look-ups~\cite{DBLP:journals/jacm/BenderCFKT23}.
    We then read $S[q+1\dd n]$ from left to right, always storing the last $4q$ letters in memory. 
    By using $\mathcal{D}$, we can assume that $S[q+1\dd n]$ consists only of letters occurring in $P$ (otherwise $S$ cannot be coverable by any prefix of $P$) and that these letters are mapped onto the range $[q]$. Every time we have a fragment $F$ from $\mathcal S'$ in memory, we compute the seeds of $F$ by employing the algorithm of \Cref{thm:fast-seeds}, which takes $\cO(q)$ time.
    In accordance with the $\cO(q)$-size representation of the computed seeds~\cite{DBLP:conf/soda/KociumakaKRRW12,DBLP:conf/esa/Radoszewski23}, we can check whether each of the $q$ candidate prefixes is a seed, by first constructing the generalized suffix tree~\cite{DBLP:conf/focs/Farach97} of 
    $P$ and $F$ in $\cO(q)$ time (thus finding which prefixes of $P$ occur in $F$), and then checking whether each of the candidate prefixes $P[1\dd i]$ which occur in $F$, for $i\in[q]$, is a seed of $F$ in $\cO(1)$ time per candidate. In addition to processing all fragments of $\mathcal S'$, we also need to check whether any of the remaining prefix candidates is a suffix and thus a border of $S$. We achieve this simply by computing the borders of string $P\$L$ in $\cO(q)$ total time~\cite{DBLP:journals/siamcomp/KnuthMP77}, where $\$$ is a letter not in $\Sigma$ and $L:=S[n-q+1 \dd n]$. The total time is thus $\cO(q \cdot n/q) = \cO(n)$.
\end{proof}

\bibliographystyle{alpha}
\bibliography{references}

\end{document}